\documentclass[journal]{new-aiaa}

\usepackage{algorithm}
\usepackage{algorithmic}
\usepackage{color}
\usepackage{stmaryrd}
\usepackage{amsfonts}
\usepackage{amssymb}
\usepackage{empheq}
\usepackage{hyperref}
\usepackage{enumerate}
\usepackage{mathrsfs}
\usepackage{empheq}
\usepackage{marginnote}

\usepackage[utf8]{inputenc}

\usepackage{graphicx}
\usepackage{subcaption}
\usepackage{amsmath}
\usepackage[version=4]{mhchem}
\usepackage{siunitx}
\usepackage{longtable,tabularx}
\newtheorem{theorem}{Theorem}

\newtheorem{lemma}{Lemma}

\newtheorem{proposition}{Proposition}

\newenvironment{proof}[1][Proof]{\begin{trivlist}
\item[\hskip \labelsep {\bfseries #1}]}{\end{trivlist}}

\title{Technical Note for "A Geodesic Approach for the Control of Tethered Quadrotors"}

\author{Tam W. Nguyen\footnote{Postdoctoral Researcher, Department of Aerospace Engineering, University of Michigan, 1320 Beal Avenue, Ann Arbor, MI 48109}}
\affil{University of Michigan, Ann Arbor, MI 48109}
\author{Marco M. Nicotra\footnote{Assistant Professor, Department of Electrical, Computer, and Energy Engineering, University of Colorado Boulder, 425 UCB, Boulder, CO 80309}}
\affil{University of Colorado Boulder, Boulder, CO 80309}
\author{Emanuele Garone \footnote{Associate Professor, Department of Control Engineering and System Analysis, Universit\'{e} libre de Bruxelles, Av. F.D. Roosevelt, 50, C.P. 165/55, Belgium}}
\affil{Universit\'{e} libre de Bruxelles, Brussels, 1050, Belgium}

\begin{document}

\maketitle

\begin{abstract}
    This technical note focuses on the control of a quadrotor unmanned aerial vehicle (UAV) tethered to the ground. The control objective is to stabilize the UAV to the desired position while ensuring that the cable remains taut at all times. A cascade control scheme is proposed. The inner loop controls the attitude of the UAV. The outer loop gives the attitude reference to the inner loop, and is designed so that (i) the gravity force is compensated, (ii) the cable is taut at all times, and (iii) the trajectory of the UAV follows the geodesic path. To prove asymptotic stability, small gain arguments are used. The control scheme is augmented with a reference governor to enforce constraints.
\end{abstract}

\section{Introduction}
\lettrine{U}{nmanned} Aerial Vehicles (UAVs) are very capable aerial platforms, and are used for surveillance, environmental interactions, and object manipulation \cite{UAV_CoopConstruction,Tiltrotor,nguyen2016control}. The potential of UAVs is still limited by factors such as flight time and onboard capabilities. A possible way to mitigate these issues is to connect the UAV to a ground station by a tether, capable of supplying energy, transmitting data, and/or applying forces.
Possible examples of tethered UAVs include: assisting the landing of a helicopter on a ship \cite{UAVTether_Helilanding}, and improving
fight stability in the presence of wind \cite{M_UAVKiteConfig}.

Since the presence of the tether influences the dynamics of the UAV, it is required to develop dedicated control strategies.
Most schemes in the literature use model inversion techniques.
In this note, we use a cascade control scheme, which does not require an accurate model to stabilize the system. 
This approach was first introduced in \cite{Taut_Automatica} for a bi-dimensional tethered UAV.

In this note, the saturation of the actuators are considered. We show that, due to the cable constraint and the saturations of the actuators, the points of equilibrium of the controlled system are only locally stable. Therefore, we augment the scheme with a Reference Governor (RG) \cite{garone2017reference} to enlarge the domain of attraction of the points of equilibrium.
It is also shown that the presence of transient in the inner loop can lead to a loss of cable tension. This behavior can be worsened in the presence of input saturations. This issue is again solved using the RG.

This technical note provides all the proofs and technicalities of the manuscript ``\emph{A Geodesic Approach for the Control of Tethered Quadrotors}". For more details on the literature, points of equilibrium, and numerical analyses, the reader is referred to the complete manuscript.

\section{Problem Statement}

Consider the 3D model of a quadrotor tethered to the ground.
We use the usual model of a UAV \cite{mayhew2009robust}, which is subject to the holonomic constraint
\begin{subequations}
\begin{empheq}[left=\empheqlbrace]{align}
m\ddot{\mathbf{p}}=& T\mathbf{R}\hat{\mathbf{z}}-mg\hat{\mathbf{z}} & \label{eq:positiondynamics}\\
\mathcal{\mathbf{J}}\dot{\pmb{\omega}}=& -\pmb{\omega}^{\wedge}\mathcal{\mathbf{J}}\pmb{\omega} + \pmb{\tau}, & \label{eq:attitudedynamics}\\
\dot{\mathbf{q}}=&\dfrac{1}{2}E(\mathbf{q})\pmb{\omega}\label{eq:quaterniondynamics},\\
\text{subject to:}\nonumber\\
\lVert \mathbf{p} \lVert = & L, \label{eq:pconstraint}
\end{empheq}
\end{subequations}
where $m\in\mathbb{R}_{> 0}$ is the mass of the UAV, $\mathbf{p}\in\mathbb{R}^3$ the position of the UAV, $T\in\mathbb{R}_{\geq 0}$ the UAV thrust, $\mathbf{R}\in SO(3)$ the UAV attitude rotation matrix, $\hat{\mathbf{z}}:=\left[\begin{matrix}0 & 0 & 1\end{matrix}\right]^T$ the vertical component of the inertial frame, $g\in\mathbb{R}_{> 0}$ the gravity acceleration, $\mathcal{\mathbf{J}}>0\in\mathbb{R}^{3\times3},\mathcal{\mathbf{J}}=\mathcal{\mathbf{J}}^T$ the moment of inertia of the UAV, $\pmb{\omega}:=\left[\begin{matrix}\omega_x & \omega_y & \omega_z \end{matrix}\right]^T\in\mathbb{R}^{3}$ the angular velocity of the UAV, $\pmb{\tau}\in\mathbb{R}^3$ the resultant torque of the UAV, $L\in\mathbb{R}_{>0}$ the length of the cable, and $\mathbf{q}:=\left[\begin{matrix}q_0 & \mathbf{q}_v^T\end{matrix}\right]^T\in\mathbb{H}$ the quaternion associated to $\mathbf{R}$ with $q_0\in\mathbb{R}$ as the real part, and $\mathbf{q}_v\in\mathbb{R}^3$ the imaginary part of $\mathbf{q}$.
$E(\mathbf{q}):=\left[\begin{matrix}-\mathbf{q}_v & q_0 \mathbf{I}_3+\mathbf{q}_v^{\wedge}\end{matrix}\right]^T\in\mathbb{R}^{4\times3}$ is the quaternion differential kinematics, $\mathbf{I}_3\in\mathbb{R}^{3\times3}$ the identity matrix, and $.^{\wedge}:\mathbb{R}^3\to\mathbb{R}^{3\times3}$ the skew operator defined as
\begin{equation}\label{eq:skewoperator}
\pmb{\omega}^{\wedge}:=\left[
\begin{matrix}
0 & -\omega_z & \omega_y \\
\omega_z & 0 & -\omega_x \\
-\omega_y & \omega_x & 0
\end{matrix}
\right].
\end{equation}

The thrust $T$ is generated by the propellers of the UAV and is aligned with the $z$-component of the body frame $\hat{\mathbf{z}}_b=\mathbf{R}\hat{\mathbf{z}}$. Furthermore, we assume that $T$ is limited and that the actuators are saturated as
\begin{eqnarray}\label{eq:Tconstraint}
0 \leq T \leq T_{max}, & T_{max} > mg.
\end{eqnarray}

It is worth noting that, $\mathbf{p}$ can be parameterized using the azimuthal angle $\phi\in[0,2\pi)$ and the polar angle $\theta\in(-\pi/2,\pi/2)$.
Moreover, (\ref{eq:pconstraint}) implies the existence of a reaction force opposite to $T_c\in\mathbb{R}_{\geq 0}$, which is the projection of the active force $\mathbf{F}_a=T\mathbf{R}\hat{\mathbf{z}}-mg\hat{\mathbf{z}}$ on the cable axis. For a massless and inextensible cable, to ensure that the cable is taut at all times, $T_c$ must satisfy\footnote{In this note, we denote the scalar product and vector product between two vectors in $\mathbb{R}^3$ as $\langle \cdot \;,\cdot \rangle:\mathbb{R}^3\times\mathbb{R}^3 \to \mathbb{R}$ and $\cdot \times \cdot :\mathbb{R}^3\times\mathbb{R}^3 \to \mathbb{R}^3$, respectively.}
\begin{equation} \label{eq:Tcconstraint}
T_c=\langle \mathbf{F}_a, \hat{\mathbf{r}} \rangle \geq T_{c,min},
\end{equation}
where $\hat{\mathbf{r}}=\mathbf{p}/L$, and $T_{c,min}\in\mathbb{R}_{\geq 0}$ is an arbitrary tension.
System (\ref{eq:positiondynamics}),(\ref{eq:pconstraint}) is equivalent to
\begin{equation}\label{eq:positiondynamicswithtension}
m\ddot{\mathbf{p}}=T\mathbf{R}\hat{\mathbf{z}}-mg\hat{\mathbf{z}}-T_c\hat{\mathbf{r}}
\end{equation}
under the assumption that $\lVert \mathbf{p}(0) \lVert = L$.

The objective is to stabilize the UAV to any desired position $\mathbf{p}_d$ such that $\lVert \mathbf{p}_d \lVert = L$, while maintaining constraint (\ref{eq:Tcconstraint}) on the cable satisfied at all times.

\section{Onboard Control}

The objective of the onboard controller is to ensure that $\lim_{t\to\infty} \mathbf{p}(t) = \mathbf{p}_d$.
A cascade control strategy is proposed, and we design the outer loop assuming an ideal inner loop.

\subsection{Ideal Attitude Dynamics}

Assume that $\mathbf{R}(t)=\mathbf{R}_d$ at each instant $t$. The system dynamics can be rewritten as
\begin{equation}\label{eq:pdynamicsideal}
\begin{cases}
m\ddot{\mathbf{p}} = T\mathbf{R}_d\hat{\mathbf{z}}-mg\hat{\mathbf{z}},\\
\text{subject to (\ref{eq:pconstraint})},
\end{cases}
\end{equation}
where $T$ and $\mathbf{R}_d$ are the control inputs.
The proposed control law for the desired thrust vector $T\mathbf{R}_d\hat{\mathbf{z}}$ is
\begin{equation}\label{eq:Tcontrollaw}
T\mathbf{R}_d\hat{\mathbf{z}}=T_t\hat{\mathbf{t}}+T_g\hat{\mathbf{z}}+T_p\hat{\mathbf{r}},
\end{equation}
where $T_t\in\mathbb{R}$ is the tangential term that we use to control the position of the UAV, $T_g=mg\in\mathbb{R}_{> 0}$ is a constant gravity compensation term, and $T_p\in(T_{c,min},T_{max}-mg)$ is a constant pulling term on the cable. 
To control $T_t\hat{\mathbf{t}}$, we use the PD control law \cite{bullo1995control},
\begin{equation}\label{eq:Ttlaw}
T_t\hat{\mathbf{t}}=\text{dist}(\mathbf{p},\mathbf{p}_d)K_{p,t}
\left[
\begin{matrix}
\langle \hat{\mathbf{t}},\hat{\mathbf{x}} \rangle \\
\langle \hat{\mathbf{t}},\hat{\mathbf{y}} \rangle \\
\langle \hat{\mathbf{t}},\hat{\mathbf{z}} \rangle
\end{matrix}
\right]
- K_{d,t} \dot{\mathbf{p}},
\end{equation}
where 
\begin{equation}\label{eq:tunit}
\begin{aligned}
\hat{\mathbf{t}}=&\dfrac{(\mathbf{p} \times \mathbf{p}_d)\times \mathbf{p}}{\max\{\lVert (\mathbf{p} \times \mathbf{p}_d)\times \mathbf{p} \lVert,\mu\}}
\end{aligned}
\end{equation}
is the unit gradient of the geodesic path ($\mu>0$),
\begin{equation}\label{eq:dist}
\text{dist}(\mathbf{p},\mathbf{p}_d):=L\arccos(\langle \mathbf{p}/L,\mathbf{p}_d/L \rangle)
\end{equation}
is the great-circle distance between $\mathbf{p}$ and $\mathbf{p}_d$, and $K_{p,t},K_{d,t}\in\mathbb{R}_{>0}$.

In the following lemma, it can be proven that (\ref{eq:pdynamicsideal}) controlled by (\ref{eq:Tcontrollaw}) and (\ref{eq:Ttlaw}) is exponentially stable considering an ideal attitude dynamics.
\begin{lemma}\label{lem:outer}
Consider System (\ref{eq:positiondynamics}),(\ref{eq:pconstraint}) controlled by (\ref{eq:Tcontrollaw}) and (\ref{eq:Ttlaw}). For an ideal attitude dynamics, the equilibrium point $(\mathbf{p},\dot{\mathbf{p}})=(\mathbf{p}_d,0)$ is exponentially stable for any initial condition satisfying
\begin{equation}\label{eq:Kpcondition}
K_{p,t}>\dfrac{\lVert \dot{\mathbf{p}}(0) \lVert^2}{\pi^2-\text{dist}(\mathbf{p}(0),\mathbf{p}_d)^2}.
\end{equation}
\end{lemma}
\begin{proof}
Using the control law (\ref{eq:Tcontrollaw}) in (\ref{eq:pdynamicsideal}), we obtain
\begin{equation}\label{eq:pdynamicscontrolled}
\begin{cases}
m\ddot{\mathbf{p}}=T_t\hat{\mathbf{t}} + T_g\hat{\mathbf{z}} + T_p\hat{\mathbf{r}} - mg\hat{\mathbf{z}},\\
\text{subject to (\ref{eq:pconstraint})}.
\end{cases}
\end{equation}
Since the attitude dynamics is ideal, $T_g\hat{\mathbf{z}}$ cancels $mg\hat{\mathbf{z}},$ and $T_p\hat{\mathbf{r}}$ is cancelled out by the reaction force $-T_c\hat{\mathbf{r}}$ at all times.
As a consequence, (\ref{eq:pdynamicscontrolled}) can be rewritten as
\begin{equation}\label{eq:pdynamicscontrolled2}
\begin{cases}
m\ddot{\mathbf{p}}=T_t\hat{\mathbf{t}},\\
\text{subject to (\ref{eq:pconstraint})}.
\end{cases}
\end{equation}
Using (\ref{eq:Ttlaw}) in (\ref{eq:pdynamicscontrolled2}), it follows from \cite[Theorem 4]{bullo1995control} that the closed loop system is exponentially stable for any initial condition satisfying (\ref{eq:Kpcondition}). It is worth noting that the stability results of the point of equilibrium are semi-global. Indeed, it follows from Eq. \eqref{eq:Kpcondition} that, for any initial position belonging to the spherical dome, there exists a sufficiently large $K_{p,t}$ such that the system trajectories will exponentially tend to $\mathbf{p}_d$. \begin{flushright}$\square$\end{flushright}
\end{proof}

Next, we compute 
the thrust $T$ and the desired rotation matrix $\mathbf{R}_d$.
First, decompose $T\mathbf{R}_d\hat{\mathbf{z}}$ as
\begin{equation}
T\mathbf{R}_d\hat{\mathbf{z}}=T_{d,x}\hat{\mathbf{x}}+T_{d,y}\hat{\mathbf{y}}+T_{d,z}\hat{\mathbf{z}},
\end{equation}
where $T_{d,x}:=T_t\langle \hat{\mathbf{t}},\hat{\mathbf{x}} \rangle + T_p \langle \hat{\mathbf{r}},\hat{\mathbf{x}} \rangle$, $T_{d,y}:=T_t\langle \hat{\mathbf{t}},\hat{\mathbf{y}} \rangle + T_p \langle \hat{\mathbf{r}},\hat{\mathbf{y}} \rangle$, and $T_{d,z}:=T_t\langle \hat{\mathbf{t}},\hat{\mathbf{z}} \rangle + T_p \langle \hat{\mathbf{r}},\hat{\mathbf{z}} \rangle + T_g$. 
Accordingly, $T$ can be computed as
\begin{equation}\label{eq:Tnorm}
T=\sqrt{T_{d,x}^2+T_{d,y}^2+T_{d,z}^2}.
\end{equation}

Concerning $\mathbf{R}_d$, consider that $\mathbf{R}_d$ is parameterized by the quaternion $\mathbf{q}_d$.
The particular solution $\mathbf{q}_{\zeta}:=\left[\begin{matrix}q_{\zeta,0} & \mathbf{q}_{\zeta,v}^T\end{matrix}\right]^T$ corresponds to the minimal rotation between $\hat{\mathbf{z}}$ and $T\mathbf{R}_d\hat{\mathbf{z}}$. Define $\zeta_d\in[-\pi,\pi)$ as the angle between $\hat{\mathbf{z}}$ and $T\mathbf{R}_d\hat{\mathbf{z}}$ by
\begin{equation}\label{eq:zeta}
\zeta_d=\arctan 2\left(\sqrt{T_{d,x}^2+T_{d,y}^2},T_{d,z}\right).
\end{equation}
The particular solution $\mathbf{q}_{\zeta}$ is computed by\footnote{If $T_{d,x}=0\wedge T_{d,y}=0$, we have $q_{\zeta,0}=1$ and $\mathbf{q}_{\zeta,v}=\left[\begin{matrix}0&0&0\end{matrix}\right]^T$.}
\begin{eqnarray}\label{eq:qzeta}
q_{\zeta,0}=&\cos{\dfrac{\zeta_d}{2}},\\
\mathbf{q}_{\zeta,v}=&\dfrac{\sin{\frac{\zeta_d}{2}}}{\sqrt{T_{d,x}^2+T_{d,y}^2}}
\left[\begin{matrix}
T_{d,y}\\
T_{d,x}\\
0
\end{matrix}\right].
\end{eqnarray}
The desired quaternion $\mathbf{q}_d$ is the combination of an arbitrary rotation $\psi\in[-\pi,\pi)$ about $\hat{\mathbf{z}}$ and the minimal rotation $\zeta_d$, that is,
\begin{equation}
\left[
\begin{matrix}
q_{d,0}\\
\mathbf{q}_{d,v}
\end{matrix}
\right]
=
\left[
\begin{matrix}
q_{\zeta,0} & -\mathbf{q}_{\zeta,v}^T\\
\mathbf{q}_{\zeta,v} & q_{\zeta,0}\mathbf{I}_3+\mathbf{q}_{\zeta,v}^{\wedge}
\end{matrix}
\right]
\left[
\begin{matrix}
q_{\psi,0}\\
\mathbf{q}_{\psi,v}
\end{matrix}
\right],
\end{equation}
where $q_{\psi,0}=\cos\dfrac{\psi}{2}$ and $\mathbf{q}_{\psi,v}=\sin\dfrac{\psi}{2}\hat{\mathbf{z}}$.

\subsection{Presence of Attitude Dynamics}

Here, we study under which conditions stability is preserved in the presence of attitude dynamics.

\subsubsection{Inner and Outer Loop Dynamics}

Define the error quaternion $\tilde{\mathbf{q}}:=\left[\begin{matrix}\tilde{q}_0 & \tilde{\mathbf{q}}_v^T\end{matrix}\right]^T$ as
\begin{equation}
\tilde{\mathbf{q}}:=\mathcal{R}^{-1}(\tilde{\mathbf{R}}(\tilde{\mathbf{q}})),
\end{equation}
where
\begin{equation}\label{eq:rotationerror}
\tilde{\mathbf{R}}(\tilde{\mathbf{q}}):=\mathbf{R}^T(\mathbf{q})\mathbf{R}_d(\mathbf{q}_d)
\end{equation}
is the attitude error and $\mathcal{R}^{-1}$ is the inverse Euler-Rodrigues operator \cite{dai2015euler}.
To control the UAV attitude, we use the PD control law
\begin{equation}\label{eq:innercontrollaw}
\pmb{\tau}=K_{p,q}\tilde{\mathbf{q}}_v-K_{d,q}\pmb{\omega},
\end{equation}
where $K_{p,q},K_{d,q}\in\mathbb{R}_{> 0}$ are positive scalars.
The inner loop attitude dynamics can be reformulated as
\begin{align}\label{eq:innerloopdynamics}
\begin{cases}
\dot{\tilde{\mathbf{q}}}=\dfrac{1}{2}E(\tilde{\mathbf{q}})(\pmb{\omega}-\pmb{\omega}_d)\\
\mathcal{\mathbf{J}}\dot{\pmb{\omega}}=-\pmb{\omega}^{\wedge}\mathcal{\mathbf{J}}\pmb{\omega}+K_{p,q}\tilde{\mathbf{q}}_v-K_{d,q}\pmb{\omega},
\end{cases}
\end{align}
where $\pmb{\omega}_d$ can be seen as an exogenous disturbance injected by the outer loop and is the rate of change of the desired attitude $\mathbf{R}_d$.

Regarding the outer loop, we isolate $\tilde{\mathbf{R}}$ from $\mathbf{R}_d$ manipulating (\ref{eq:rotationerror}) as
\begin{equation}\label{eq:Rexplicit}
\mathbf{R}=\mathbf{R}_d + \mathbf{R}_d (\tilde{\mathbf{R}}^T - \mathbf{I}_3).
\end{equation}
It is worth noting that $\mathbf{R}_d (\tilde{\mathbf{R}}^T - \mathbf{I}_3)$ tends to zero when $\tilde{\mathbf{R}}^T \to \mathbf{I}_3$. As a consequence, we can rewrite $T\mathbf{R}\hat{\mathbf{z}}$ as
\begin{equation}\label{eq:Tsplit}
T\mathbf{R}\hat{\mathbf{z}} = T\mathbf{R}_d \hat{\mathbf{z}} + T \mathbf{R}_d (\tilde{\mathbf{R}}^T - \mathbf{I}_3)\hat{\mathbf{z}}.
\end{equation}
Then, using (\ref{eq:Tcontrollaw}), (\ref{eq:Ttlaw}), and (\ref{eq:Tsplit}) in (\ref{eq:positiondynamics}), the outer loop dynamic can be rewritten as
\begin{equation}\label{eq:outerloop}
\begin{cases}
m\ddot{\mathbf{p}}=T_t\hat{\mathbf{t}} + \pmb{\delta}_{\tilde{\zeta}},\\
\text{subject to (\ref{eq:pconstraint})},
\end{cases}
\end{equation}
where $\pmb{\delta}_{\tilde{\zeta}}:=T \mathbf{R}_d(\tilde{\mathbf{R}}^T-\mathbf{I}_3) \hat{\mathbf{z}}$ can be seen as an exogenous disturbance injected by the inner loop dynamics.
The following lemma proves that $\pmb{\delta}_{\tilde{\zeta}}$ can be bounded by a function of class-$\mathcal{K}$ in $\tilde{\zeta}$, where $\tilde{\zeta}$ is the angle associated to the error quaternion $\tilde{\mathbf{q}}$.
\begin{lemma}\label{lem:Kfunction}
The norm of the exogenous input $\lVert \pmb{\delta}_{\tilde{\zeta}} \lVert$ is bounded by the class-$\mathcal{K}$ function $\lVert \pmb{\delta}_{\tilde{\zeta}} \lVert \leq \sqrt{6}T | \tilde{\zeta} |$.
\end{lemma}
\begin{proof}
The norm $\lVert \pmb{\delta}_{\tilde{\zeta}} \lVert$ is bounded by
\begin{equation}\label{eq:deltanorm}
\lVert \pmb{\delta}_{\tilde{\zeta}} \lVert \leq T \lVert \mathbf{R}_d \lVert \lVert \tilde{\mathbf{R}}_{I_z} \lVert,
\end{equation}
where $\tilde{\mathbf{R}}_{I_z}:=(\tilde{\mathbf{R}}-\mathbf{I}_3)\hat{z}$ is the last column of $\tilde{\mathbf{R}}-\mathbf{I}_3$ and $\lVert \mathbf{R}_d \lVert =1$ by definition. As a consequence, (\ref{eq:deltanorm}) becomes
\begin{equation}
\lVert \pmb{\delta}_{\tilde{\zeta}} \lVert \leq T \lVert \tilde{\mathbf{R}}_{I_z} \lVert.
\end{equation}
Next, we use
\begin{eqnarray}\label{eq:axisrepresentation}
\begin{cases}
\tilde{q}_0:=\cos\left(\frac{\tilde{\zeta}}{2}\right)\\
\tilde{\mathbf{q}}_v:=\sin\left(\frac{\tilde{\zeta}}{2}\right)\left[\begin{matrix}a_x \\ a_y \\ a_z\end{matrix}\right],
\end{cases}
\end{eqnarray}
where $a_x$, $a_y$, and $a_z$ are the $x$, $y$, and $z$-components of the normalized axis of rotation, respectively.
Developing the last column of $\tilde{\mathbf{R}}_{I_z}$ using (\ref{eq:axisrepresentation}), we obtain
\begin{align}\label{eq:deltanorm2}
\begin{split}
\lVert \pmb{\delta}_{\tilde{\zeta}} \lVert & \leq T \sqrt{(a_x a_z (1-\cos\tilde{\zeta})-a_y\sin\tilde{\zeta})^2+(a_y a_z(1-\cos\tilde{\zeta})+a_x\sin\tilde{\zeta})^2+(\cos\tilde{\zeta}+a_z^2(1-\cos\tilde{\zeta})-1))^2}\\
& \leq T\sqrt{(a_z^4-2a_z^2+1+a_x^2a_z^2+a_y^2a_z^2)(1-\cos\tilde{\zeta})^2+(a_x^2+a_y^2)\sin^2\tilde{\zeta}}.
\end{split}
\end{align}
Then, since $a_x\leq 1$, $a_y\leq 1$, and $a_z\leq 1$, \eqref{eq:deltanorm2} is upperbounded by 
\begin{equation}\label{eq:deltanorm3}
\lVert \pmb{\delta}_{\tilde{\zeta}} \lVert \leq T \sqrt{4(1-\cos^2|\tilde{\zeta}|)+2\sin^2|\tilde{\zeta}}|:=f(|\tilde{\zeta}|).
\end{equation}
Note that $f(0)=0$ and that $\dfrac{\partial f(|\tilde{\zeta}|)}{\partial \tilde{\zeta}}$ is maximal when $|\tilde{\zeta}|=0$. 
Therefore, (\ref{eq:deltanorm3}) is upperbounded by the linear class-$\mathcal{K}$ function
\begin{equation}
\lVert \pmb{\delta}_{\tilde{\zeta}} \lVert \leq \sqrt{6}T|\tilde{\zeta}|,
\end{equation}
which concludes the proof. \begin{flushright}$\square$\end{flushright}
\end{proof}

\subsubsection{Stability Properties}

The following proposition proves that the inner loop is input-to-state stable (ISS) with respect to $\pmb{\omega}_d,$ and the asymptotic gain can be made arbitrarily small.
\begin{proposition}\label{lem:innerloop}
Consider the inner loop (\ref{eq:innerloopdynamics}). Then, given $K_d\propto\sqrt{K_{p,q}}$, the system is ISS with respect to the disturbance $\pmb{\omega}_d$ and there exists an asymptotic gain $\gamma_{in}$ between $\pmb{\omega}_d$ and $\tilde{\zeta}$, which can be made arbitrarily small for sufficiently large $K_{p,q}$. 
\end{proposition}
\begin{proof}
The proof is detailed in Appendix \ref{proof:innerloop}.
\end{proof}

Concerning the outer loop, the following proposition proves that the outer loop is ISS with restriction with respect to $\tilde{\zeta}$ and that the asymptotic gain is finite.
\begin{proposition}\label{lem:outerloop}
Under the assumption $||\mathbf{p}(t)||=L$ at all times, given a desired position $\mathbf{p}_d$, System (\ref{eq:outerloop}) is ISS with restriction $|\tilde{\zeta}|<\tilde{\zeta}_{max}$ and $|T| \leq T_{sup}$ with respect to $\tilde{\zeta}$.
Furthermore, the asymptotic gain $\gamma_{out}$ between the disturbance $\tilde{\zeta}$ and $\pmb{\omega}_d$ exists and is finite.
\end{proposition}
\begin{proof}
The details of the proof can be found in Appendix \ref{proof:outerloop}.
\end{proof}

Combining Propositions \ref{lem:innerloop} and \ref{lem:outerloop}, it is possible to prove that the overall system is AS.
\begin{theorem}\label{thm:smallgain}
Consider the overall system (\ref{eq:innerloopdynamics}) and (\ref{eq:outerloop}) and assume the cable rigid.
Then, given $K_{d,q}\propto\sqrt{K_{p,q}}$, the point of equilibrium $\mathbf{p}_d$ is AS for suitably large $K_{p,q}$. 
\end{theorem}
\begin{proof}
From Propositions \ref{lem:innerloop} and \ref{lem:outerloop}, $\gamma_{in}$ and $\gamma_{out}$ are proven to be finite under the assumption $|\tilde{\zeta}|<\tilde{\zeta}_{max}$. Since $\gamma_{in}$ can be made arbitrarily small for sufficiently large $K_{p,q}$, it is always possible to ensure $\gamma_{in}\gamma_{out}<1$. Therefore, the Small Gain Theorem \cite{khalil1996noninear} can be applied and, since there exists a suitable set of initial conditions containing the equilibrium in its interior and such that  $\lVert \tilde{\zeta}\lVert_\infty<\tilde{\zeta}_{max}$, $
\lVert T \lVert_\infty \leq T_{max},$ the point of equilibrium is asymptotically stable.
\end{proof}

The next section illustrates how to increase the set of admissible initial conditions by using a Reference Governor to manage the transient response of the closed-loop system.

\section{Constraint Enforcement}

The classical discrete-time RG \cite{bemporad1998reference} computes the next applied reference at step $k$ as
\begin{equation}\label{eq:classicalRG}
	(\mathbf{p}_a)_{k+1}=L\dfrac{(1-c)(\mathbf{p}_a)_k + c\mathbf{p}_d}{\|(1-c)(\mathbf{p}_a)_k + c\mathbf{p}_d\|},
\end{equation}
where the scalar $c\in[0,1]$ is maximized over a sufficiently long prediction time horizon $t_h$.

The Explicit Reference Governor (ERG) \cite{marcoERG2018} uses the differential equation
\begin{equation}\label{eq:ERG}
	\dot{\mathbf{p}}_a=\Delta(\mathbf{p}_a,\mathbf{p},\dot{\mathbf{p}},\mathbf{R},\pmb{\omega})\pmb{\rho}(\mathbf{p}_a,\mathbf{p}_d),
\end{equation}
where $\pmb{\rho}(\mathbf{p}_a,\mathbf{p}_d)=\dfrac{(\mathbf{p}_a\times \mathbf{p}_d)\times \mathbf{p}_a}{\max\{\| (\mathbf{p}_a\times \mathbf{p}_d)\times \mathbf{p}_a \|,\eta\}}$ is an AF constructed on the gradient of the geodesics, with $\eta\in\mathbb{R}_{>0}$ as a parameter to be tuned, and
$\Delta(\mathbf{p}_a,\mathbf{p},\dot{\mathbf{p}},\mathbf{R},\pmb{\omega})=\kappa(\hat{T}_{c,m}(\mathbf{p}_a,\mathbf{p},\dot{\mathbf{p}},\mathbf{R},\pmb{\omega})-T_{c,min}+\epsilon)^2$ is the DSM that ensures constraint \eqref{eq:Tcconstraint}, with $\kappa,\epsilon\in\mathbb{R}_{> 0}$ as parameters to be tuned.

For both RG and ERG, the time horizon $t_h$ should be chosen sufficiently long so as to catch the most relevant part of the transient.
According to the recursive feasibility property of the RG and ERG, the closed-loop system augmented with the reference governor is guaranteed to reach any feasible set-point without violating the system constraints.

\section{Conclusions}

This technical note proposes a control framework to study the stabilization of tethered Unmanned Aerial Vehicles (UAVs) in three dimensions. The constraint on the cable is modeled by a holonomic constraint and is conditioned by the positiveness of the tension in the cable. A cascade control strategy is developed with the dual objective of controlling the UAV and guaranteeing the taut cable condition. Small Gain arguments are used to prove asymptotic stability of the system. The control law is augmented with the Reference Governor (RG) to enforce constraints satisfaction at all times, and enlarge the domain of attraction of the points of equilibrium. 

\bibliography{CAT-AVIATOR}


\appendix

\section{Proof of Proposition \ref{lem:innerloop}} \label{proof:innerloop}

Consider the Lyapunov function candidate
\begin{equation}\label{eq:Vin}
	\begin{aligned}
		V(\tilde{q}_0,\tilde{\mathbf{q}}_v,\pmb{\omega}) = & 2 K_{p,q} (1 - \tilde{q}_0)  + \dfrac{1}{2} 
		\left[
			\begin{matrix}
				\tilde{\mathbf{q}}_v \\
				\pmb{\omega}
			\end{matrix}
		\right]^T
		\left[
			\begin{matrix}
				4 \eta K_{d,q}\mathbf{I}_3 & 2\eta\mathcal{\mathbf{J}} \\
				2\eta\mathcal{\mathbf{J}} & \mathcal{\mathbf{J}}
			\end{matrix}
		\right]
		\left[
			\begin{matrix}
				\tilde{\mathbf{q}}_v \\
				\pmb{\omega}
			\end{matrix}
		\right],
	\end{aligned}
\end{equation}
where 
\begin{equation}\label{eq:etaquat}
	0<\eta<\min\left\{K_{d,q}\mathcal{\mathbf{J}}^{-1},\dfrac{2K_{p,q}K_{d,q}}{4\lambda_M(\mathcal{\mathbf{J}})K_{p,q}+K_{d,q}^2}\right\} \in \mathbb{R}_{>0},
\end{equation}
is a strictly positive parameter with $\lambda_M(\mathcal{\mathbf{J}})$ denoting the maximum eigenvalue of $\mathcal{\mathbf{J}}$. 
The Lyapunov candidate \eqref{eq:Vin} is positive definite since $|\tilde{q}_0| \leq 1$, and the second term is also positive for $\eta$ satisfying \eqref{eq:etaquat}. 
Moreover, note that the point of equilibrium $(\tilde{q}_0, \tilde{\mathbf{q}}_v, \pmb{\omega})=(1, 0, 0)$ gives $V(1, 0 ,0) =0$.

The time derivative of \eqref{eq:Vin} is computed by
\begin{equation}\label{eq:dotVin}
		\dot{V}(\cdot) =  - 2 K_{p,q}\dot{\tilde{q}}_0 + 4 \eta (\tilde{\mathbf{q}}_v)^TK_{d,q}\dot{\tilde{\mathbf{q}}}_v + 2 \eta \pmb{\omega}^T \mathcal{\mathbf{J}} \dot{\tilde{\mathbf{q}}}_v
		+ 2 \eta (\mathcal{\mathbf{J}} \tilde{\mathbf{q}}_v)^T \dot{\pmb{\omega}} + (\mathcal{\mathbf{J}}\pmb{\omega})^T\dot{\pmb{\omega}},
\end{equation}
which is composed of five terms that will be treated separately in the following.
Injecting the inner loop dynamics into \eqref{eq:dotVin}, we obtain the following properties: 
\begin{itemize}
	\item \textbf{Term 1}: The first term can be rewritten as 
		\begin{equation}\label{eq:fstterm}
			- 2 K_{p,q}\dot{\tilde{q}}_0 = (\tilde{\mathbf{q}}_v)^T K_{p,q} (\pmb{\omega}-\pmb{\omega}_D).
		\end{equation}
	\item \textbf{Term 2}: The second term is computed by 
		\begin{equation}\label{eq:secondterm}
			\begin{aligned}
				4 \eta (\tilde{\mathbf{q}}_v)^T K_{d,q} \tilde{\mathbf{q}}_v \dot{\tilde{\mathbf{q}}}_v = &  (\tilde{\mathbf{q}}_v)^T (2 \eta K_{d,q} (\tilde{q}_0\mathbf{I}_3 + E(\tilde{\mathbf{q}}_v)) (\pmb{\omega} - \pmb{\omega}_D) \\
				= & (\tilde{\mathbf{q}}_v)^T (2 \eta K_{d,q} \tilde{q}_0\mathbf{I}_3) (\pmb{\omega}-\pmb{\omega}_D),
			\end{aligned}
		\end{equation}
		where the last line has been derived using $(\tilde{\mathbf{q}}_v)^T(E(\tilde{\mathbf{q}}_v)(\pmb{\omega}-\pmb{\omega}_D)) = (\tilde{\mathbf{q}}_v)^T(\tilde{\mathbf{q}}_v \times (\pmb{\omega}-\pmb{\omega}_D)) =0$, according to the property $a^T(a\times b)=b^T(a\times a)=0$;
	\item \textbf{Term 3}: The third term is computed by
		\begin{equation}\label{eq:thirdterm}
			\begin{aligned}
				2 \eta \pmb{\omega}^T \mathcal{\mathbf{J}} \dot{\tilde{\mathbf{q}}}_v 
				= & \pmb{\omega}^T(\eta \mathcal{\mathbf{J}} (\tilde{q}_0\mathbf{I}_3+E(\tilde{\mathbf{q}}_v))) (\pmb{\omega} - \pmb{\omega}_D).
			\end{aligned}
		\end{equation}
	\item \textbf{Term 4}: The fourth term can be rewritten using the fact that $\mathcal{\mathbf{J}}$ is symmetric as
		\begin{equation}
			\begin{aligned}
				2\eta(\mathcal{\mathbf{J}} \tilde{\mathbf{q}}_v)^T \dot{\pmb{\omega}} = & 2\eta (\tilde{\mathbf{q}}_v)^T \mathcal{\mathbf{J}} \dot{\pmb{\omega}}\\
				= & 2 \eta (\tilde{\mathbf{q}}_v)^T(-E(\pmb{\omega})(\mathcal{\mathbf{J}} \pmb{\omega}) - K_{p,q} \tilde{\mathbf{q}}_v - K_{d,q}\pmb{\omega}),
			\end{aligned}
		\end{equation}
		where $\mathcal{\mathbf{J}}\dot{\pmb{\omega}}$ has been substituted with the inner loop dynamics.
		Then, using the fact that $a^T(b\times c) = c^T(a\times b)$, it implies that:
		\begin{equation}
			\begin{aligned}
				2\eta(\mathcal{\mathbf{J}} \tilde{\mathbf{q}}_v)^T \dot{\pmb{\omega}}  = & 2 \eta (- (\mathcal{\mathbf{J}}\pmb{\omega})^T(E(\tilde{\mathbf{q}}_v) \pmb{\omega}) - (\tilde{\mathbf{q}}_v)^T K_{p,q} \tilde{\mathbf{q}}_v - (\tilde{\mathbf{q}}_v)^T K_{d,q} \pmb{\omega}) \\
				= & - \pmb{\omega}^T (2\eta\mathcal{\mathbf{J}} E(\tilde{\mathbf{q}}_v)) \pmb{\omega} - (\tilde{\mathbf{q}}_v)^T (2\eta K_{p,q}) \tilde{\mathbf{q}}_v - (\tilde{\mathbf{q}}_v)^T(2\eta K_{d,q})\pmb{\omega},
			\end{aligned}
		\end{equation}
		and since $\pmb{\omega}^TE(\tilde{\mathbf{q}}_v)\pmb{\omega}=\pmb{\omega}^T(\tilde{\mathbf{q}}_v\times \pmb{\omega})=\pmb{\omega}^T(-\pmb{\omega}\times \tilde{\mathbf{q}}_v)=(\tilde{\mathbf{q}}_v)^T(-\pmb{\omega}\times\pmb{\omega})=0$, it follows that
		\begin{equation}\label{eq:fourthterm}
			\begin{split}
				2\eta(\mathcal{\mathbf{J}} \tilde{\mathbf{q}}_v)^T \dot{\pmb{\omega}} 
				= - (\tilde{\mathbf{q}}_v)^T (2\eta K_{p,q}) \tilde{\mathbf{q}}_v - (\tilde{\mathbf{q}}_v)^T(2\eta K_{d,q})\pmb{\omega}.
			\end{split}
		\end{equation}
	\item \textbf{Term 5}: The last term, according to the fact that $\mathcal{\mathbf{J}}$ is symmetric, can be calculated by
		\begin{equation}\label{eq:fifthterm1}
			\begin{aligned}
				(\mathcal{\mathbf{J}} \pmb{\omega})^T \dot{\pmb{\omega}} = & \pmb{\omega}^T (\mathcal{\mathbf{J}} \dot{\pmb{\omega}}) \\
				= & \pmb{\omega}^T(-E(\pmb{\omega})(\mathcal{\mathbf{J}} \pmb{\omega}) - K_{p,q}\tilde{\mathbf{q}}_v - K_{d,q} \pmb{\omega}),
			\end{aligned}
		\end{equation}
		where $\mathcal{\mathbf{J}}\dot{\pmb{\omega}}$ is substituted with the system dynamics.
		Then, using the fact that $a^T(a \times b)=b^T(a \times a) = 0$, it results
		\begin{equation}\label{eq:fifthterm}
			\begin{aligned}
				(\mathcal{\mathbf{J}} \pmb{\omega})^T \dot{\pmb{\omega}} = & -(\tilde{\mathbf{q}}_v)^T K_{p,q} \pmb{\omega} - \pmb{\omega}^T K_{d,q} \pmb{\omega}.
			\end{aligned}
		\end{equation}
\end{itemize}

Therefore, regrouping the relations \eqref{eq:fstterm}, \eqref{eq:secondterm}, \eqref{eq:thirdterm}, \eqref{eq:fourthterm}, and \eqref{eq:fifthterm} in matrix form, we obtain
\begin{equation}\label{eq:dotVin1}
	\begin{aligned}
		\dot{V}(\cdot) = &
		-
		\left[
			\begin{matrix}
				\tilde{\mathbf{q}}_v \\ \pmb{\omega}
			\end{matrix}
		\right]^T
		\left[
			\begin{matrix}
				2 \eta K_{p,q} \mathbf{I}_3 & \eta K_{d,q}(\mathcal{I}_{3}-h^R\mathbf{I}_3) \\
				\eta K_{d,q}(\mathcal{I}_{3}-\tilde{q}_0\mathbf{I}_3) & K_{d,q}\mathbf{I}_3 - \eta\mathcal{\mathbf{J}}(\tilde{q}_0\mathbf{I}_3+E(\tilde{\mathbf{q}}_v))
			\end{matrix}
		\right]
		\left[
			\begin{matrix}
				\tilde{\mathbf{q}}_v \\ \pmb{\omega}
			\end{matrix}
		\right] 
		-
		\left[
			\begin{matrix}
				\tilde{\mathbf{q}}_v \\ \pmb{\omega}
			\end{matrix}
		\right]^T
		\left[
			\begin{matrix}
				(K_{p,q}+2\eta K_{d,q}\tilde{q}_0)\mathbf{I}_3 \\
				\eta\mathcal{\mathbf{J}}(\tilde{q}_0\mathbf{I}_3 + E(\tilde{\mathbf{q}}_v))
			\end{matrix}
		\right]
		\pmb{\omega}_D.
	\end{aligned}
\end{equation}

Next, we use the angle-axis representationto make the error angle $\tilde{\zeta}$ appear in the equations.  
Doing so, it is possible to upper-bound the time derivative $\dot{V}(\cdot)$ of \eqref{eq:dotVin1} as
\begin{equation}\label{eq:dotVin2}
	\begin{aligned}
		\dot{V}(\cdot) \leq &
		-
		\left[
			\begin{matrix}
				\lVert\sin\frac{\tilde{\zeta}}{2}\lVert \\ \lVert\pmb{\omega}\lVert
			\end{matrix}
		\right]^T
		Q_{in}(\tilde{\zeta})
		\left[
			\begin{matrix}
				\lVert\sin\frac{\tilde{\zeta}}{2}\lVert \\
				\lVert\pmb{\omega}\lVert
			\end{matrix}
		\right]
		+
		\left[
			\begin{matrix}
				\lVert\sin\frac{\tilde{\zeta}}{2}\lVert \\ \lVert\pmb{\omega}\lVert
			\end{matrix}
		\right]^T
		D_{in}(\tilde{\zeta})
		\lVert \pmb{\omega}_D \lVert_{\infty},
	\end{aligned}
\end{equation}
where
\begin{align}
	Q_{in}(\tilde{\zeta}) = &
	\left[
		\begin{matrix}
			2 \eta K_{p,q} & \eta K_{d,q}\left(1-\left\lVert\cos\frac{\tilde{\zeta}}{2}\right\lVert\right) \\
			\eta K_{d,q}\left(1 - \left\lVert\cos\frac{\tilde{\zeta}}{2}\right\lVert\right) & K_{d,q} - \eta \lambda_M(\mathcal{\mathbf{J}})\left(\left\lVert\cos\frac{\tilde{\zeta}}{2}\right\lVert + \left\lVert \sin\frac{\tilde{\zeta}}{2} \right\lVert\right)
		\end{matrix}
	\right]\in\mathbb{R}^{2 \times 2}, \\
	D_{in}(\tilde{\zeta}) = &
	\left[
		\begin{matrix}
			K_{p,q} + 2\eta K_{d,q}\left\lVert\cos\frac{\tilde{\zeta}}{2}\right\lVert \\
			\eta \lambda_M(\mathcal{\mathbf{J}}) \left( \left\lVert\cos \frac{\tilde{\zeta}}{2}\right\lVert + \left\lVert \sin\frac{\tilde{\zeta}}{2} \right\lVert \right)
		\end{matrix}
	\right]\in\mathbb{R}^{2 \times 1}.
\end{align}

Since $0\leq\left\lVert\cos\frac{\tilde{\zeta}}{2}\right\lVert \leq 1$, and $0\leq \left\lVert\sin\frac{\tilde{\zeta}}{2} \right\lVert \leq 1$, we can lower-bound $Q_{in}(\tilde{\zeta})$ with
\begin{equation}
	Q_{in}(\tilde{\zeta})  \geq \bar{Q}_{in} =
	\left[
		\begin{matrix}
			2 \eta K_{p,q} & \eta K_{d,q} \\
			\eta K_{d,q} & K_{d,q}-2\eta\lambda_M(\mathcal{\mathbf{J}})
		\end{matrix}
	\right] \in \mathbb{R}^{2 \times 2},
\end{equation}
which is positive definite for $\eta$ chosen such that it satisfies \eqref{eq:etaquat}. 
Similarly, it is possible to upper-bound $D_{in}(\tilde{\zeta})$ by
\begin{equation}
	D_{in}(\tilde{\zeta}) \leq \bar{D}_{in} =
	\left[
		\begin{matrix}
			K_{p,q} + 2 \eta K_{d,q}  \\
			2\eta\lambda_M(\mathcal{\mathbf{J}})
		\end{matrix}
	\right] \in \mathbb{R}^{2 \times 1}.
\end{equation}

Hence, $\dot{V}(\cdot)$ in \eqref{eq:dotVin2} satisfies
\begin{equation}\label{eq:dotVin3}
	\begin{split}
		\dot{V}(\cdot) \leq 
		-
		\left[
			\begin{matrix}
				\lVert\sin\frac{\tilde{\zeta}}{2}\lVert \\ \lVert\pmb{\omega} \lVert
			\end{matrix}
		\right]^T
		\bar{Q}_{in}
		\left[
			\begin{matrix}
				\lVert\sin\frac{\tilde{\zeta}}{2}\lVert \\ \lVert\pmb{\omega}\lVert 
			\end{matrix}
		\right]
		+
		\left[
			\begin{matrix}
				\lVert\sin\frac{\tilde{\zeta}}{2}\lVert\\ 
				\lVert\pmb{\omega}\lVert
			\end{matrix}
		\right]^T
		\bar{D}_{in}
		\lVert \pmb{\omega}_D \lVert_{\infty},
	\end{split}
\end{equation}
where $\bar{Q}_{in}$ and $\bar{D}_{in}$ are constant matrices.
As a result, it is possible to derive the following implication:
\begin{equation}
	\left[
		\begin{matrix}
			\lVert \sin\frac{\tilde{\zeta}}{2} \lVert \\ \lVert \pmb{\omega} \lVert
		\end{matrix}
	\right] >
	\left(\bar{Q}_{in}^{-1}\bar{D}_{in}\right)\lVert\pmb{\omega}_D\lVert_{\infty} \Rightarrow \dot{V}<0,
\end{equation}
which makes the system ISS with respect to the exogenous input $\pmb{\omega}_D$.

It remains to prove that the asymptotic gain $\gamma_{in}$ between $\tilde{\zeta}$ and $\pmb{\omega}_D$ can be made arbitrarily small.
This asymptotic gain is given by
\begin{equation}
	\left[ \begin{matrix} \gamma_{in} \\ \cdot \end{matrix} \right] = \left(\bar{Q}_{in}^{-1}\bar{D}_{in}\right) \lVert \omega_D \lVert_{\infty}.
\end{equation}

Considering the parameter choice $K_{d,q}\propto \sqrt{K_{p,q}}$, the following proportional dependencies are derived using the dominant degree of $K_{p,q}$ in each element of the matrices
\begin{align}
	\bar{Q}_{in} \propto &
	\left[
		\begin{matrix}
			K_{p,q} & \sqrt{K_{p,q}} \\
			\sqrt{K_{p,q}} & \sqrt{K_{p,q}}
		\end{matrix}
	\right],\\
	\dfrac{1}{\text{det}(\bar{Q}_{in})} \propto & \dfrac{1}{K_{p,q}\sqrt{K_{p,q}}}, \\
	\bar{Q}_{in}^{-1} \propto &
	\left[
		\begin{matrix}
			\frac{1}{K_{p,q}} & \frac{1}{K_{p,q}} \\
			\frac{1}{K_{p,q}} & \frac{1}{\sqrt{K_{p,q}}}
		\end{matrix}
	\right],\\
	\bar{D}_{in} \propto &
	\left[
		\begin{matrix}
			K_{p,q} \\
			\cdot
		\end{matrix}
	\right].
\end{align}

From the above-statements, it follows
\begin{align}
	\bar{Q}_{in}^{-1}\bar{D}_{in} \propto 
	\left[
		\begin{matrix}
			\frac{1}{K_{p,q}} \\
			\frac{1}{\sqrt{K_{p,q}}}
		\end{matrix}
	\right] \Rightarrow
	\gamma_{in} \propto  \frac{1}{K_{p,q}},
\end{align}
which concludes the proof.

\section{Proof of Proposition 2}\label{proof:outerloop}

Define $B_p:=\{\hat{\pmb{\theta}},\hat{\pmb{\phi}}\}$ as the orthonormal basis for $\mathbb{S}^2:=\text{span}\{\mathbf{p}\}^\perp$, where $\mathbb{S}^2$ is the field of vectors tangent to the surface of a sphere of radius $L$.
The system dynamics can be rewritten as
\begin{align}\label{eq:outerloop2}
\begin{cases}
\dot{\mathbf{p}}=& v_\theta\hat{\pmb{\theta}} + v_\phi\hat{\pmb{\phi}}\\
\dot{\mathbf{v}}=& \text{dist}(\mathbf{p},\mathbf{p}_d)h_{p,t}
\left[
\begin{matrix}
\langle \hat{\textbf{t}},\hat{\pmb{\theta}} \rangle \\
\langle \hat{\textbf{t}},\hat{\pmb{\phi}} \rangle
\end{matrix}
\right]
- h_{d,t} \mathbf{v} + \pmb{\Delta}_{\tilde{\zeta}},
\end{cases}
\end{align}
where $\mathbf{v}:=[v_\theta,v_\phi]^T\in\mathbb{R}^{2}$ is the velocity vector whose components $v_{\theta}$ and $v_{\phi}$ are the polar and azimuthal velocities, respectively, $h_{p,t}:=K_{p,t}/m \in\mathbb{R}_{> 0}$ and $h_{d,t}:=K_{d,t}/m\in\mathbb{R}_{> 0}$ the proportional and the derivative gains divided by $m$, respectively, and $\pmb{\Delta}_{\tilde{\zeta}}$ the projected exogenous input
\begin{equation}
\pmb{\Delta}_{\tilde{\zeta}}:=
\left[
\begin{matrix}
\langle \delta_{\tilde{\zeta}},\hat{\pmb{\theta}}\rangle\\
\langle \delta_{\tilde{\zeta}},\hat{\pmb{\phi}}\rangle
\end{matrix}
\right] \in \mathbb{R}^2.
\end{equation}

Next, consider the Lyapunov function candidate \cite{bullo1995control}
\begin{equation}\label{eq:lyapunovBullo}
V_{out}=\dfrac{1}{2}\text{dist}(\mathbf{p},\mathbf{p}_d)^2 + \dfrac{h_{p,t}^{-1}}{2}\lVert \mathbf{v} \lVert^2 - \epsilon\text{Cross},
\end{equation}
where $\epsilon\in\mathbb{R}_{>0}$ is a positive parameter such that $\epsilon<\sqrt{h^{-1}_{p,t}},$ and
\begin{equation}\label{eq:cross}
\text{Cross}:=\text{dist}(\mathbf{p},\mathbf{p}_d)\left\langle \mathbf{v}, \left[
\begin{matrix}
\langle \hat{\textbf{t}},\hat{\pmb{\theta}} \rangle \\
\langle \hat{\textbf{t}},\hat{\pmb{\phi}} \rangle
\end{matrix}
\right]
\right\rangle.
\end{equation}
The time derivative of (\ref{eq:lyapunovBullo}) is
\begin{equation}\label{eq:Vdotstep}
\dot{V}_{out}=\text{dist}(\mathbf{p},\mathbf{p}_d)\dfrac{d}{dt}\left(\text{dist}(\mathbf{p},\mathbf{p}_d)\right) + h_{p,t}^{-1}\lVert \mathbf{v} \lVert \langle \dfrac{\mathbf{v}}{\lVert \mathbf{v} \lVert}, \dot{\mathbf{v}} \rangle - \epsilon\dfrac{d}{dt}\text{Cross}.
\end{equation}

In order to make the following steps clearer, we split (\ref{eq:Vdotstep}) into two parts
\begin{equation}\label{eq:Vdot}
\dot{V}_{out}=\dot{V}_1 - \dot{V}_2,
\end{equation}
where
\begin{align}
\dot{V}_1:=&\text{dist}(\mathbf{p},\mathbf{p}_d)\dfrac{d}{dt}\left(\text{dist}(\mathbf{p},\mathbf{p}_d)\right) + h_{p,t}^{-1}\langle \mathbf{v}, \dot{\mathbf{v}}\rangle \label{eq:V1}\\
\dot{V}_2:=&\epsilon\dfrac{d}{dt}\text{Cross} \label{eq:V2}.
\end{align}

Computing the time derivative of the great-circle distance, we obtain
\begin{equation}\label{eq:lemBullo}
\dfrac{d}{dt}(\text{dist}(\mathbf{p},\mathbf{p}_d))=-\langle \dot{\mathbf{p}},\hat{\textbf{t}} \rangle
\end{equation}
because $\dot{\mathbf{p}}\perp \mathbf{p}$ and therefore $\langle \dot{\mathbf{p}},\mathbf{p} \rangle=0$.
As a consequence, using (\ref{eq:lemBullo}) in (\ref{eq:V1}), we obtain
\begin{equation}\label{eq:V1_2}
\dot{V}_1=-\text{dist}(\mathbf{p},\mathbf{p}_d)\langle \dot{\mathbf{p}},\hat{\textbf{t}} \rangle + h_{p,t}^{-1}\langle \mathbf{v},\dot{\mathbf{v}}\rangle.
\end{equation}
Injecting the dynamics (\ref{eq:outerloop2}) in (\ref{eq:V1_2}), we obtain
\begin{equation}\label{eq:V1_3}
\dot{V}_1=
-\text{dist}(\mathbf{p},\mathbf{p}_d)\left\langle \mathbf{v},\left[\begin{matrix}
\langle \hat{\textbf{t}},\hat{\pmb{\theta}}\rangle \\
\langle \hat{\textbf{t}},\hat{\pmb{\phi}}\rangle
\end{matrix}
\right]
\right\rangle + h_{p,t}^{-1} \left \langle \mathbf{v}, \left(
\text{dist}(\mathbf{p},\mathbf{p}_d) h_{p,t}\left[
\begin{matrix}
\langle \hat{\textbf{t}},\hat{\pmb{\theta}} \rangle \\
\langle \hat{\textbf{t}},\hat{\pmb{\phi}} \rangle
\end{matrix}
\right]
- h_{d,t} \mathbf{v}
\right) \right\rangle + h_{p,t}^{-1} \langle v, \pmb{\Delta}_{\tilde{\zeta}} \rangle.
\end{equation}
Then, after simplifications, (\ref{eq:V1_2}) becomes
\begin{equation}
\dot{V}_1=-h_{p,t}^{-1} h_{d,t} \lVert \mathbf{v} \lVert^2 + h_{p,t}^{-1} \langle \mathbf{v},\pmb{\Delta}_{\tilde{\zeta}} \rangle,
\end{equation}
which is upper-bounded by
\begin{equation}\label{eq:V1final}
\dot{V}_1\leq-h_{p,t}^{-1} h_{d,t} \lVert \mathbf{v} \lVert^2 + h_{p,t}^{-1} \lVert \mathbf{v} \lVert \lVert \pmb{\Delta}_{\tilde{\zeta}} \lVert.
\end{equation}

For what regards $V_2$, we can rewrite (\ref{eq:V2}) as
\begin{equation}\label{eq:V2_2}
\dot{V}_2=\epsilon\left\{
\dfrac{d}{dt}(\text{dist}(\mathbf{p},\mathbf{p}_d))\langle \mathbf{v},\left[
\begin{matrix}
\langle \hat{\textbf{t}},\hat{\pmb{\theta}}\rangle \\
\langle \hat{\textbf{t}},\hat{\pmb{\phi}}\rangle
\end{matrix}
\right] \rangle+\text{dist}(\mathbf{p},\mathbf{p}_d)\left(
\langle \dot{\mathbf{v}} , \left[
\begin{matrix}
\langle \hat{\textbf{t}},\hat{\pmb{\theta}} \rangle\\
\langle \hat{\textbf{t}},\hat{\pmb{\phi}} \rangle
\end{matrix}
\right]\rangle
+ \langle \mathbf{v},\dfrac{d}{dt}\hat{\textbf{t}}\rangle \right)
\right\}.
\end{equation}
Using (\ref{eq:lemBullo}) and injecting the dynamics (\ref{eq:outerloop2}) in (\ref{eq:V2_2}), we obtain
\begin{equation}\label{eq:V2final}
\dot{V}_2=\epsilon(-\lVert \mathbf{v} \lVert^2 + B_1 + B_2),
\end{equation}
where
\begin{equation}\label{eq:B1}
B_1:=
\text{dist}(\mathbf{p},\mathbf{p}_d)\left\langle \text{dist}(\mathbf{p},\mathbf{p}_d) h_{p,t}
\left[
\begin{matrix}
\langle \hat{\textbf{t}},\hat{\pmb{\theta}} \rangle\\
\langle \hat{\textbf{t}},\hat{\pmb{\phi}} \rangle
\end{matrix}
\right]
- h_{d,t}\mathbf{v}+\pmb{\Delta}_{\tilde{\zeta}},
\left[
\begin{matrix}
\langle \hat{\textbf{t}},\hat{\pmb{\theta}} \rangle\\
\langle \hat{\textbf{t}},\hat{\pmb{\phi}} \rangle
\end{matrix}
\right]\right\rangle,
\end{equation}
and
\begin{equation}\label{eq:B2}
B_2:=\text{dist}(\mathbf{p},\mathbf{p}_d)\langle \mathbf{v},\dfrac{d}{dt}\hat{\textbf{t}} \rangle.
\end{equation}

For what concerns $B_1$, developing the scalar product in (\ref{eq:B1}) leads to
\begin{equation}\label{eq:B1_2}
B_1=
\text{dist}(\mathbf{p},\mathbf{p}_d)^2h_{p,t}-\text{dist}(\mathbf{p},\mathbf{p}_d)h_{d,t}
\langle \mathbf{v},\left[
\begin{matrix}
\langle \hat{\textbf{t}},\hat{\pmb{\theta}}\rangle \\
\langle \hat{\textbf{t}},\hat{\pmb{\phi}}\rangle
\end{matrix}
\right]\rangle+\text{dist}(\mathbf{p},\mathbf{p}_d)
\langle \pmb{\Delta}_{\tilde{\zeta}},\left[
\begin{matrix}
\langle \hat{\textbf{t}},\hat{\pmb{\theta}}\rangle \\
\langle \hat{\textbf{t}},\hat{\pmb{\phi}}\rangle
\end{matrix}
\right]\rangle.
\end{equation}
As a consequence, we can lower-bound (\ref{eq:B1_2}) with
\begin{equation}\label{eq:B1final}
B_1\geq h_{p,t} \text{dist}(\mathbf{p},\mathbf{p}_d)^2-h_{d,t}\text{dist}(\mathbf{p},\mathbf{p}_d)\lVert \mathbf{v} \lVert - \text{dist}(\mathbf{p},\mathbf{p}_d)\lVert \pmb{\Delta}_{\tilde{\zeta}} \lVert.
\end{equation}

The next step is to compute $B_2$. To do so, we first transform $\hat{\textbf{t}}$ using the following manipulation
\begin{equation}\label{eq:tscalar}
\hat{\textbf{t}}=\dfrac{\mathbf{p}_d-\langle \mathbf{p}/L,\mathbf{p}_d/L \rangle \mathbf{p}}{L\sin\Delta_{\sigma}},
\end{equation}
where
\begin{equation}\label{eq:deltasigma}
\Delta_{\sigma}:=\text{dist}(\mathbf{p},\mathbf{p}_d)/L
\end{equation}
is the angle between $\mathbf{p}$ and $\mathbf{p}_d$.
The time derivative of (\ref{eq:tscalar}) is
\begin{equation}\label{eq:ddtt}
\dfrac{d}{dt}\hat{\textbf{t}}=\dfrac{(-\langle \dot{\mathbf{p}}/L,\mathbf{p}_d/L \rangle \mathbf{p} - \langle \mathbf{p}/L,\mathbf{p}_d/L \rangle \dot{\mathbf{p}})\sin\Delta_{\sigma}}{L\sin^2\Delta_{\sigma}}-\dfrac{(\mathbf{p}_d-\langle \mathbf{p}/L,\mathbf{p}_d/L \rangle \mathbf{p})\cos\Delta_{\sigma}\dot{\Delta}_{\sigma}}{L\sin^2\Delta_{\sigma}}.
\end{equation}
Following from (\ref{eq:deltasigma}) and (\ref{eq:lemBullo}), we have
\begin{equation}\label{eq:deltasigmadot}
\dot{\Delta}_{\sigma}=-\dfrac{\langle \dot{\mathbf{p}},\hat{\textbf{t}} \rangle}{L}.
\end{equation}
Then, re-using (\ref{eq:tscalar}) and injecting (\ref{eq:deltasigmadot}) in (\ref{eq:ddtt}), we obtain
\begin{equation}\label{eq:ddtt2}
\dfrac{d}{dt}\hat{\textbf{t}}=
-\dfrac{\langle \dot{\mathbf{p}}/L,\mathbf{p}_d/L \rangle \mathbf{p}}{L\sin\Delta_{\sigma}}-\dfrac{\langle \mathbf{p}/L,\mathbf{p}_d/L \rangle \dot{\mathbf{p}}}{L\sin\Delta_{\sigma}}+\dfrac{\cos\Delta_{\sigma}}{\sin\Delta_{\sigma}}\hat{\textbf{t}}\dfrac{\langle \dot{\mathbf{p}},\hat{\textbf{t}} \rangle}{L}.
\end{equation}
At this point, we can use $\langle \mathbf{p}/L,\mathbf{p}_d/L \rangle=\cos\Delta_{\sigma}$ in (\ref{eq:ddtt2}) as
\begin{equation}\label{eq:ddttexplicit}
\dfrac{d}{dt}\hat{\textbf{t}}=-\dfrac{\langle \dot{\mathbf{p}}/L,\mathbf{p}_d/L \rangle \mathbf{p}}{L\sin\Delta_{\sigma}} - \dfrac{\cos\Delta_{\sigma}}{L\sin\Delta_{\sigma}}(\dot{\mathbf{p}}-\hat{\textbf{t}}\langle \dot{\mathbf{p}},\hat{\textbf{t}} \rangle).
\end{equation}
Since $\dot{\mathbf{p}}\perp \mathbf{p}$, we have $\langle \dot{\mathbf{p}},\mathbf{p} \rangle=0$ and using (\ref{eq:deltasigma}) and (\ref{eq:ddttexplicit}) in (\ref{eq:B2}), we obtain
\begin{equation}\label{eq:B2interm}
B_2=-\dfrac{\Delta_{\sigma}\cos\Delta_{\sigma}}{\sin\Delta_{\sigma}}(\lVert \dot{\mathbf{p}} \lVert^2 - \langle \dot{\mathbf{p}},\hat{\textbf{t}} \rangle^2).
\end{equation}
Remark that, since $\max_{\Delta_{\sigma}\in(-\pi/2,\pi/2)}\left\{\dfrac{\Delta_{\sigma}\cos\Delta_{\sigma}}{\sin\Delta_{\sigma}}\right\}=1$, (\ref{eq:B2interm}) can be lower-bounded by
\begin{equation}\label{eq:B2final}
B_2\geq -\lVert \dot{\mathbf{p}} \lVert^2.
\end{equation}

Consequently, combining (\ref{eq:Vdot}), (\ref{eq:V1final}), (\ref{eq:V2final}), (\ref{eq:B1final}), and (\ref{eq:B2final}), we can upper-bound $\dot{V}_{out}$ as
\begin{equation}\label{eq:Vdotineq}
\dot{V}_{out}\leq-h_{p,t}^{-1}h_{d,t}\lVert  \mathbf{v} \lVert^2 + \epsilon(2\lVert \mathbf{v} \lVert^2- h_{p,t}\text{dist}(\mathbf{p},\mathbf{p}_d)^2+h_{d,t}\text{dist}(\mathbf{p},\mathbf{p}_d)\lVert \mathbf{v} \lVert)+(h_{p,t}^{-1}\lVert \mathbf{v} \lVert + \epsilon\text{dist}(\mathbf{p},\mathbf{p}_d))\lVert \pmb{\Delta}_{\tilde{\zeta}} \lVert.
\end{equation}
Note that the great-circle distance can be bounded by
\begin{equation}
\lVert \tilde{\mathbf{p}} \lVert \leq \text{dist}(\mathbf{p},\mathbf{p}_d)<\dfrac{\pi}{2} \lVert \tilde{\mathbf{p}} \lVert,
\end{equation}
where $\tilde{\mathbf{p}}:=\mathbf{p}-\mathbf{p}_d$.
As a consequence, we can rewrite (\ref{eq:Vdotineq}) as
\begin{eqnarray}\label{eq:Voutmatrix}
\dot{V}_{out}\leq -
[\begin{matrix}
\lVert \tilde{\mathbf{p}} \lVert & \lVert \mathbf{v} \lVert
\end{matrix}]
Q
\left[
\begin{matrix}
\lVert \tilde{\mathbf{p}} \lVert \\
\lVert \mathbf{v} \lVert
\end{matrix}
\right]+h_{p,t}^{-1}\lVert \mathbf{v} \lVert \lVert \pmb{\Delta}_{\tilde{\zeta}} \lVert + \epsilon\lVert \tilde{\mathbf{p}} \lVert\lVert \pmb{\Delta}_{\tilde{\zeta}} \lVert,
\end{eqnarray}
where 
\begin{equation}
Q=\left[
\begin{matrix}
\epsilon h_{p,t} & -\epsilon\dfrac{h_{d,t}\pi}{4} \\
-\epsilon\dfrac{h_{d,t}\pi}{4} & h_{p,t}^{-1}h_{d,t}-2\epsilon
\end{matrix}
\right].
\end{equation}
The first term on the right-hand side of (\ref{eq:Voutmatrix}) is strictly negative if $Q$ is positive definite meaning if
\begin{equation}
\epsilon <\frac{16h_{d,t}}{32h_{p,t} +h_{d,t}^2\pi^2}.
\end{equation}
Under this condition, $\dot{V}_{out}$ is negative definite if $(\tilde{\mathbf{p}},\mathbf{v})$ satisfies
\begin{equation}\label{eq:ISSbound}
\left\|
\begin{bmatrix}
\tilde{\mathbf{p}} \\
\mathbf{v}
\end{bmatrix}\right\| > \mu^{-1} \left\|\begin{bmatrix}
\epsilon \\
h_{p,t}^{-1}
\end{bmatrix}\right\| \lVert \pmb{\Delta}_{\tilde{\zeta}} \lVert,
\end{equation}
where $\mu$ is the smallest eigenvalue of the matrix $Q$.
Then, due to the presence of the cable, it is worth noting that $\|\tilde{\mathbf{p}}\|\leq2L$. As a result, it follows from equation \eqref{eq:ISSbound} that $\|\pmb{\Delta}_{\tilde\zeta}\|$ must be upper bounded by
\begin{equation}
\Delta_{\tilde{\zeta},\max{}} = \frac{4\mu L^2}{\sqrt{\epsilon^2+h_{p,t}^{-2}}}.
\end{equation}
Since the origin is ISS with restrictions on $\Delta_{\tilde\zeta}$, it follows from Lemma \ref{lem:Kfunction}, that it is also ISS with restrictions on $\tilde{\zeta}$, which concludes the proof.

\end{document}